\documentclass{amsart}
\usepackage{mathrsfs}
\usepackage{cases}
\usepackage{amssymb,amscd,amsthm}
\usepackage{latexsym}
\usepackage{hyperref}
\usepackage{color}

\usepackage[noadjust]{cite}

\usepackage{mathtools}

\usepackage{subfigure}
\usepackage{graphicx}

\date{\today}

\newcommand{\Z}{{\mathbb Z}}
\newcommand{\R}{{\mathbb R}}
\newcommand{\C}{{\mathbb C}}
\newcommand{\D}{{\mathbb D}}
\newcommand{\T}{{\mathbb T}}

\newcommand{\N}{{\mathbb N}}

\newcommand{\set}[1]{\left\{#1\right\}}
\newcommand{\abs}[1]{\left| #1\right|}
\newcommand{\norm}[1]{\left\|#1\right\|}

\newcommand{\hdist}{\mathrm{dist}_\mathrm{H}}

\newtheorem{theorem}{Theorem} [section]
\newtheorem{remark}[theorem]{Remark}
\newtheorem{lemma}[theorem]{Lemma}
\newtheorem{prop}[theorem]{Proposition}
\newtheorem{coro}[theorem]{Corollary}
\newtheorem{definition}[theorem]{Definition}
\newtheorem{remarks}[theorem]{Remarks}
\newtheorem{claim}[theorem]{Claim}

\begin{document}

\title[OPUC with Fibonacci Verblunsky Coefficients]{Orthogonal Polynomials on the Unit Circle with Fibonacci Verblunsky Coefficients,\\ II.~Applications}

\author[D.\ Damanik]{David Damanik}

\address{Department of Mathematics, Rice University, Houston, TX~77005, USA}

\email {\href{mailto:damanik@rice.edu}{damanik@rice.edu}}

\urladdr {\href{http://www.ruf.rice.edu/~dtd3}{www.ruf.rice.edu/$\sim$dtd3}}

\author[P.\ Munger]{Paul Munger}

\address{Department of Mathematics, Rice University, Houston, TX~77005, USA}

\email{\href{mailto:pem1@rice.edu}{pem1@rice.edu}}

\urladdr{\href{http://pem1.web.rice.edu/}{http://pem1.web.rice.edu/}}

\author[W.\ N.\ Yessen]{William N.\ Yessen}

\address{Department of Mathematics, University of California, Irvine, CA~92697, USA}

\email{\href{mailto:wyessen@math.uci.edu}{wyessen@math.uci.edu}}

\urladdr{\href{http://sites.google.com/site/wyessen/}{http://sites.google.com/site/wyessen/}}

\thanks{D.\ D.\ was supported in part by a Simons Fellowship and NSF grant DMS--1067988.}

\thanks{W.\ N.\ Y.\ was supported by NSF grant DMS-0901627, PI: A. Gorodetski}

\keywords{orthogonal polynomials, quantum walks, Fibonacci sequence}
\subjclass[2000]{Primary 42C05; Secondary 37D99}

\begin{abstract}
We consider CMV matrices with Verblunsky coefficients determined in an appropriate way by the Fibonacci sequence and present two applications of the spectral theory of such matrices to problems in mathematical physics. In our first application we estimate the spreading rates of quantum walks on the line with time-independent coins following the Fibonacci sequence. The estimates we obtain are explicit in terms of the parameters of the system. In our second application, we establish a connection between the classical nearest neighbor Ising model on the one-dimensional lattice in the complex magnetic field regime, and CMV operators. In particular, given a sequence of nearest-neighbor interaction couplings, we construct a sequence of Verblunsky coefficients, such that the support of the Lee-Yang zeros of the partition function for the Ising model in the thermodynamic limit coincides with the essential spectrum of the CMV matrix with the constructed Verblunsky coefficients. Under certain technical conditions, we
also show that the zeros distribution measure coincides with the density of states measure for the CMV matrix.
\end{abstract}

\maketitle

\section{Introduction}\label{s.intro}

A CMV matrix is a semi-infinite matrix of the form
$$
\mathcal{C} = \begin{pmatrix}
{}& \bar\alpha_0 & \bar\alpha_1 \rho_0 & \rho_1
\rho_0
& 0 & 0 & \dots & {} \\
{}& \rho_0 & -\bar\alpha_1 \alpha_0 & -\rho_1
\alpha_0
& 0 & 0 & \dots & {} \\
{}& 0 & \bar\alpha_2 \rho_1 & -\bar\alpha_2 \alpha_1 &
\bar\alpha_3 \rho_2 & \rho_3 \rho_2 & \dots & {} \\
{}& 0 & \rho_2 \rho_1 & -\rho_2 \alpha_1 &
-\bar\alpha_3
\alpha_2 & -\rho_3 \alpha_2 & \dots & {} \\
{}& 0 & 0 & 0 & \bar\alpha_4 \rho_3 & -\bar\alpha_4
\alpha_3
& \dots & {} \\
{}& \dots & \dots & \dots & \dots & \dots & \dots & {}
\end{pmatrix}
$$
where $\alpha_n \in \D = \{ w \in \C : |w| < 1 \}$ and $\rho_n = (1-|\alpha_n|^2)^{1/2}$. $\mathcal{C}$ defines a unitary operator on $\ell^2(\Z_+)$.

CMV matrices $\mathcal{C}$ are in one-to-one correspondence to probability measures $\mu$ on the unit circle $\partial \D$ that are not supported by a finite set. To go from $\mathcal{C}$ to $\mu$, one invokes the spectral theorem. To go from $\mu$ to $\mathcal{C}$, one can proceed either via orthogonal polynomials or via Schur functions. In the approach via orthogonal polynomials, the $\alpha_n$'s arise as recursion coefficients for the polynomials.

Explicitly, consider the Hilbert space $L^2(\partial \D,d\mu)$ and apply the Gram-Schmidt orthonormalization procedure to the
sequence of monomials $1 , w , w^2 , w^3 , \ldots$. This yields a sequence $\varphi_0 , \varphi_1 , \varphi_2 , \varphi_3 , \ldots$ of normalized polynomials that are pairwise orthogonal in $L^2(\partial \D,d\mu)$. Corresponding to $\varphi_n$, consider the ``reflected polynomial'' $\varphi_n^*$, where the coefficients of $\varphi_n$ are conjugated and then written in reverse order. Then, we have
\begin{equation}\label{e.tmbasic}
\begin{pmatrix} \varphi_{n+1}(w) \\ \varphi_{n+1}^*(w) \end{pmatrix} = \rho_n^{-1} \left( \begin{array}{cc} w & - \bar{\alpha}_n \\ - \alpha_n w& 1 \end{array} \right) \begin{pmatrix} \varphi_{n}(w) \\ \varphi_{n}^*(w) \end{pmatrix}
\end{equation}
for suitably chosen $\alpha_n \in \D$ (and again with $\rho_n = (1-|\alpha_n|^2)^{1/2}$). With these $\alpha_n$'s, one may form the corresponding CMV matrix $\mathcal{C}$ as above and obtain a unitary matrix for which the spectral measure corresponding to the cyclic vector $\delta_0$ is indeed the measure $\mu$ we based the construction on.

Sometimes it makes sense to consider extended CMV matrices, acting on $\ell^2(\Z)$. They have essentially the same form, but are two-sided infinite and may be put in one-to-one correspondence with two-sided infinite sequences $\{ \alpha_n \}_{n \in \Z} \subset \D$. They are typically denoted by $\mathcal{E}$. Hence, such an extended CMV matrix, corresponding to a sequence $\{ \alpha_n \}_{n \in \Z} \subset \D$, takes the form
$$
\mathcal{E} = \begin{pmatrix}
{}& \dots & \dots & \dots & \dots & \dots & \dots & \dots & \dots & {} \\
{}& \dots & -\bar\alpha_{-3} \alpha_{-4} & -\rho_{-3} \alpha_{-4} & 0 & 0 & 0 & 0 & \dots & {} \\
{}& \dots & \bar\alpha_{-2} \rho_{-3} & -\bar\alpha_{-2} \alpha_{-3} & \bar\alpha_{-1} \rho_{-2} & \rho_{-1} \rho_{-2} & 0 & 0 & \dots & {} \\
{}& \dots & \rho_{-2} \rho_{-3} & -\rho_{-2} \alpha_{-3} & -\bar\alpha_{-1} \alpha_{-2} & -\rho_{-1} \alpha_{-2} & 0 & 0 & \dots & {} \\
{}& \dots & 0 & 0 & \bar\alpha_0 \rho_{-1} & -\bar\alpha_0 \alpha_{-1} & \bar\alpha_1 \rho_0 & \rho_1 \rho_0 & \dots & {} \\
{}& \dots & 0 & 0 & \rho_0 \rho_{-1} & -\rho_0 \alpha_{-1} & -\bar\alpha_1 \alpha_0 & -\rho_1 \alpha_0 & \dots & {} \\
{}& \dots & 0 & 0 & 0 & 0 & \bar\alpha_2 \rho_1 & -\bar\alpha_2 \alpha_1 & \dots & {} \\
{}& \dots & \dots & \dots & \dots & \dots & \dots & \dots & \dots & {}
\end{pmatrix}
$$
For example, when the $\alpha_n$'s are obtained by sampling along the orbit of an invertible ergodic transformation, the general theory of such operators naturally considers the two-sided situation. Special cases of this scenario that are of great interest include the periodic case, the almost periodic case, and the random case.

\bigskip

The present paper is a continuation of \cite{DMY}. In this series we focus on Verblunsky coefficients that are generated by the Fibonacci substitution. This is also a special case of the general ergodic situation mentioned in the previous paragraph. The theory behind the Fibonacci case is rich and appealing. It has strong connections to non-trivial questions in dynamical systems and, due to this connection, one is able to establish very fine results for operators arising in this way. In part one of this series, \cite{DMY}, we mainly focused on the local structure of the spectrum and in particular related the local dimension at a given energy in the spectrum to the value of the Fricke-Vogt invariant at the corresponding initial point of the trace map. In this paper we consider two problems in mathematical physics, which can be related to a CMV matrix whose Verblunsky coefficients are determined by an element of the subshift associated with the Fibonacci substitution.

The first application, presented in Section~\ref{s.qw}, concerns a quantum walk in a time-independent Fibonacci environment. Quantum walks have been actively studied recently; compare, for example, \cite{ACMSWW, AVWW, CGMV, CGMV2, J11, J12, JM, KS, VS}. The existing literature has focused on cases where the underlying unitary operator has either purely absolutely continuous spectrum or pure point spectrum. On the other hand, the connection between quantum walks and CMV matrices exhibited in \cite{CGMV} works in full generality and applies in principle to cases with non-trivial (or purely) singular continuous spectrum. The challenge, however, is that operators with singular continuous spectrum are difficult to analyze. One reason behind this is that singular continuous spectral measures may have dimensions anywhere from zero to one, and hence the associated system may behave like one with pure point spectrum or one with purely absolutely continuous spectrum, or anywhere in between these extremes. In Section~\
ref{s.qw} we carry out an analysis of a model with purely singular continuous spectrum. In fact, not only does this model have purely singular continuous spectrum, the choice of the coins (through a Fibonacci sequence) is interesting and physically relevant, and it leads to interesting mathematics. Explicitly, we will consider a quantum walk on the line with Fibonacci coins and prove estimates for the rates of spreading of a given initial state. We work out analogs of results that had been obtained earlier for the time-evolution of the Schr\"odinger equation with a Fibonacci potential. We emphasize that our application demonstrates the full strength of the CGMV method from \cite{CGMV}.

The second application, presented in Section~\ref{sec:Ising}, is to the classical one-dimensional nearest neighbor Ising model. In this case we consider the model immersed in the complex magnetic field and study the zeros of the associated partition function, as a function of the magnetic field. We begin with the model on a finite lattice and consider the thermodynamic limit (as the size of the lattice goes to infinity). Due to the famous Lee-Yang theorem, the zeros lie on the unit circle and their distribution in the thermodynamic limit is a physically relevant and typically a nontrivial problem. One of the parameters of the model is the choice of the nearest neighbor interaction couplings. In case the sequence of interactions is constant (or periodic), the model admits an explicit solution (i.e., the zeros can be computed explicitly even in the thermodynamic limit), due to the fact that the associated transfer matrices commute and hence can be diagonalized simultaneously (the partition function on a
lattice of size $N$ is expressed as the trace of the transfer matrix over $N$ sites; or, equivalently, the trace of the product of transfer matrices over sites $1, 2, \dots, N$). In the realm of quasicrystals, a physically interesting case is that of quasi-periodic (say, Fibonacci) couplings. This problem has been considered in a few places (see \cite{Baake1995, Barata2001} and references therein). The first rigorous results in this direction appeared recently in \cite{Y3}; until then the results consisted of numerical (and quite accurate, in fact, as our rigorous results confirmed) computations and some soft (but nontrivial) analysis. In \cite{Y3}, the previously postulated multifractal structure of the zeros in the thermodynamic limit of the Ising partition function with Fibonacci-modulated couplings was proved; however, an open problem remained: \textit{how does the above result depend on the choice of a sequence from the subshift generated by the Fibonacci substitution sequence?} One way to attack this
problem would be to relate it to a spectral problem of some operator (such as a Jacobi operator or a CMV matrix). A spectral-theoretic approach to the problem has been postulated in a number of works, e.g. \cite{Baake1995, Barata2001}, but to the best of our knowledge, no connection with any class of operators has hitherto been established. Upon closer investigation, we discovered that to every nearest-neighbor Ising model in complex magnetic field, there can be associated (constructively, in fact) a CMV matrix whose essential spectrum, provided certain technical conditions hold, would coincide with the (smallest closed set containing all) Lee-Yang zeros in the thermodynamic limit; moreover, if the essential spectrum of the associated CMV matrix is not all of the unit circle, then the counting probability measures on the zeros of the partition function on the finite lattice converge weakly (as the size of the lattice goes to infinity) to a measure supported on the set of zeros in the thermodynamic limit,
which is precisely the density of states measure for the corresponding CMV matrix. These results are quite general, and applicable to the quasi-periodic case at hand, allowing us to apply all the powerful tools from spectral theory to the classical Ising models; in particular, we prove that the distribution in the thermodynamic limit of the Lee-Yang zeros is independent of the choice of the sequence (of nearest neighbor couplings) from the subshift generated by the Fibonacci substitution sequence. Let us mention that this connection is quite peculiar in the following sense. Due to the Lieb-Schultz-Mattis ansatz, the \textit{quantum} one-dimensional Ising model is equivalent to a Jacobi operator via a certain canonical transformation. We now also know that in some sense (see the results of Section~\ref{sec:Ising}) the \textit{classical} nearest-neighbor Ising model is equivalent to the CMV matrices. On the other hand, Jacobi operators are to orthogonal polynomials on the real line as the CMV matrices are to
orthogonal polynomials on the unit circle; and often results proved for one are also proved for the other.

\section{Quantum Walk in a Fibonacci Environment}\label{s.qw}

In this section we study quantum walks on the line with coins that are time-independent and follow a Fibonacci sequence in the space variable. This choice is of clear interest since the Fibonacci sequence is the central example of a quasi-crystalline structure in one space dimension. Our goal is to show how quickly an initially localized quantum state must spread out in space under the quantum walk evolution given this choice of coins.

\subsection{Quantum Walks on the Line}

Let us recall the standard quantum walk formalism. The Hilbert space is given by $\mathcal{H} = \ell^2(\Z) \otimes \C^2$. A basis is given by the elementary tensors $|n \rangle \otimes | \! \uparrow \rangle$, $|n \rangle \otimes | \! \downarrow \rangle$, $n \in \Z$. A quantum walk scenario is given as soon as coins
$$
C_{n,t} = \begin{pmatrix} c^{11}_{n,t} & c^{12}_{n,t} \\ c^{21}_{n,t} & c^{22}_{n,t} \end{pmatrix} \in U(2), \quad n,t \in \Z
$$
are specified. As one passes from time $t$ to time $t+1$, the update rule of the quantum walk is as follows,
\begin{align*}
|n \rangle \otimes | \! \uparrow \rangle & \mapsto c^{11}_{n,t} |n+1 \rangle \otimes | \! \uparrow \rangle + c^{21}_{n,t} |n-1 \rangle \otimes | \! \downarrow \rangle, \\
|n \rangle \otimes | \! \downarrow \rangle & \mapsto c^{12}_{n,t} |n+1 \rangle \otimes | \! \uparrow \rangle + c^{22}_{n,t} |n-1 \rangle \otimes | \! \downarrow \rangle.
\end{align*}
(Extend this by linearity to general elements of $\mathcal{H}$.)

There are two cases of special interest:
\begin{itemize}

\item time-homogeneous: $C_{n,t} = C_n$ for every $t \in \Z$,

\item space-homogeneous: $C_{n,t} = C_t$ for every $n \in \Z$.

\end{itemize}

We will focus our attention on the time-homogeneous case. That is, we will assume that coins $C_n \in U(2)$, $n \in \Z$ are given. In this case, there is a fixed unitary operator $U : \mathcal{H} \to \mathcal{H}$ that provides the update mechanism, and hence the powers of this unitary operator provide the time evolution. In particular, spectral theory enters the game as the powers of a unitary operator are most conveniently studied with the help of the spectral theorem.

\subsection{The CGMV Connection}

Consider the time-homogeneous situation. Thus, the coins are given by
\begin{equation}\label{e.timehomocoins}
C_{n} = \begin{pmatrix} c^{11}_{n} & c^{12}_{n} \\ c^{21}_{n} & c^{22}_{n} \end{pmatrix} \in U(2), \quad n \in \Z
\end{equation}
and the update rule, for every time transition $t \mapsto t + 1$, is given by
\begin{align}
|n \rangle \otimes | \! \uparrow \rangle & \mapsto c^{11}_{n} |n+1 \rangle \otimes | \! \uparrow \rangle + c^{21}_{n} |n-1 \rangle \otimes | \! \downarrow \rangle, \label{e.updaterule1} \\
|n \rangle \otimes | \! \downarrow \rangle & \mapsto c^{12}_{n} |n+1 \rangle \otimes | \! \uparrow \rangle + c^{22}_{n} |n-1 \rangle \otimes | \! \downarrow \rangle \label{e.updaterule2},
\end{align}
which defines the unitary operator $U : \mathcal{H} \to \mathcal{H}$ by linearity.

A suitable choice of the order of the standard basis of $\mathcal{H}$ reveals that the associated matrix representation of $U$ looks very much like an extended CMV matrix. This useful observation is due to Cantero, Gr\"unbaum, Moral, and Vel\'azquez \cite{CGMV}.

Consider the following order of the basis elements:
\begin{equation}\label{e.orderedbasis}
\ldots, |-1 \rangle \otimes | \! \uparrow \rangle , \, |-1 \rangle \otimes | \! \downarrow \rangle , \, |0 \rangle \otimes | \! \uparrow \rangle , \, |0 \rangle \otimes | \! \downarrow \rangle , \, |1 \rangle \otimes | \! \uparrow \rangle , \, |1 \rangle \otimes | \! \downarrow \rangle , \ldots .
\end{equation}
In this basis, the matrix representation of $U : \mathcal{H} \to \mathcal{H}$ is given by
\begin{equation}\label{e.umatrixrep}
U = \begin{pmatrix}
{}& \dots & \dots & \dots & \dots & \dots & \dots & \dots & \dots & {} \\
{}& \dots & 0 & c^{12}_{-2} & 0 & 0 & 0 & 0 & \dots & {} \\
{}& \dots & c^{21}_{-1} & 0 & 0 & c^{11}_{-1} & 0 & 0 & \dots & {} \\
{}& \dots & c^{22}_{-1} & 0 & 0 & c^{12}_{-1} & 0 & 0 & \dots & {} \\
{}& \dots & 0 & 0 & c^{21}_{0} & 0 & 0 & c^{11}_{0} & \dots & {} \\
{}& \dots & 0 & 0 & c^{22}_{0} & 0 & 0 & c^{12}_{0} & \dots & {} \\
{}& \dots & 0 & 0 & 0 & 0 & c^{21}_{1} & 0 & \dots & {} \\
{}& \dots & \dots & \dots & \dots & \dots & \dots & \dots & \dots & {}
\end{pmatrix},
\end{equation}
as can be checked readily using the update rule \eqref{e.updaterule1}--\eqref{e.updaterule2}, which gives rise to the unitary operator $U$; compare \cite[Section~4]{CGMV}.\footnote{Note that we follow the conventions of \cite{CGMV} here. One could argue that the correct matrix to consider is the transpose of $U$ in \eqref{e.umatrixrep}. To conform with \cite{CGMV} and subsequent papers, we will consider the matrix $U$ as given above in what follows. For our results, this does not make a difference since our matrix entries will be real and hence the transpose of $U$ is the inverse of $U$. Since our argument is based on spectral continuity of $U$, and the spectral continuity properties of $U$ and $U^{-1}$ are the same, the final result does not depend on the choice one makes at this juncture.}

Recall that an extended CMV matrix corresponding to Verblunsky coefficients $\{ \alpha_n \}_{n \in \Z}$ has the form
$$
\mathcal{E} = \begin{pmatrix}
{}& \dots & \dots & \dots & \dots & \dots & \dots & \dots & \dots & {} \\
{}& \dots & -\bar\alpha_{-3} \alpha_{-4} & -\rho_{-3} \alpha_{-4} & 0 & 0 & 0 & 0 & \dots & {} \\
{}& \dots & \bar\alpha_{-2} \rho_{-3} & -\bar\alpha_{-2} \alpha_{-3} & \bar\alpha_{-1} \rho_{-2} & \rho_{-1} \rho_{-2} & 0 & 0 & \dots & {} \\
{}& \dots & \rho_{-2} \rho_{-3} & -\rho_{-2} \alpha_{-3} & -\bar\alpha_{-1} \alpha_{-2} & -\rho_{-1} \alpha_{-2} & 0 & 0 & \dots & {} \\
{}& \dots & 0 & 0 & \bar\alpha_0 \rho_{-1} & -\bar\alpha_0 \alpha_{-1} & \bar\alpha_1 \rho_0 & \rho_1 \rho_0 & \dots & {} \\
{}& \dots & 0 & 0 & \rho_0 \rho_{-1} & -\rho_0 \alpha_{-1} & -\bar\alpha_1 \alpha_0 & -\rho_1 \alpha_0 & \dots & {} \\
{}& \dots & 0 & 0 & 0 & 0 & \bar\alpha_2 \rho_1 & -\bar\alpha_2 \alpha_1 & \dots & {} \\
{}& \dots & \dots & \dots & \dots & \dots & \dots & \dots & \dots & {}
\end{pmatrix}.
$$
In particular, if all Verblunsky coefficients with odd index vanish, the matrix becomes (recall that $\rho_n = (1-|\alpha_n|^2)^{1/2}$)
\begin{equation}\label{e.ecmvoddzero}
\mathcal{E} = \begin{pmatrix}
{}& \dots & \dots & \dots & \dots & \dots & \dots & \dots & \dots & {} \\
{}& \dots & 0 & - \alpha_{-4} & 0 & 0 & 0 & 0 & \dots & {} \\
{}& \dots & \bar\alpha_{-2} & 0 & 0 & \rho_{-2} & 0 & 0 & \dots & {} \\
{}& \dots & \rho_{-2} & 0 & 0 & -\alpha_{-2} & 0 & 0 & \dots & {} \\
{}& \dots & 0 & 0 & \bar\alpha_0 & 0 & 0 & \rho_0 & \dots & {} \\
{}& \dots & 0 & 0 & \rho_0 & 0 & 0 & -\alpha_0 & \dots & {} \\
{}& \dots & 0 & 0 & 0 & 0 & \bar\alpha_2 & 0 & \dots & {} \\
{}& \dots & \dots & \dots & \dots & \dots & \dots & \dots & \dots & {}
\end{pmatrix}.
\end{equation}

The matrix in \eqref{e.ecmvoddzero} strongly resembles the matrix representation of $U$ in \eqref{e.umatrixrep}. Note, however, that all $\rho_n$'s need to be real and non-negative for genuine CMV matrices, and this property is not guaranteed when matching \eqref{e.umatrixrep} and \eqref{e.ecmvoddzero}. But this can be easily resolved, as shown in \cite{CGMV}. Concretely, given $U$ as in \eqref{e.umatrixrep}, write
$$
c_n^{kk} = |c_n^{kk}| e^{i \sigma^k_n}, \quad n \in \Z, \; k \in \{ 1,2 \}, \; \sigma^k_n \in [0,2\pi)
$$
and define $\{ \lambda_n \}_{n \in \Z}$ by
\begin{align*}
\lambda_0 & = 1 \\
\lambda_{-1} & = 1 \\
\lambda_{2n+2} & = e^{-i \sigma^1_n} \lambda_{2n} \\
\lambda_{2n+1} & = e^{i \sigma^2_n} \lambda_{2n-1}.
\end{align*}
With the unitary matrix $\Lambda = \mathrm{diag}(\ldots, \lambda_{-1} , \lambda_0 , \lambda_1 , \ldots)$, we then have
$$
\mathcal{E} = \Lambda^* U \Lambda,
$$
where $\mathcal{E}$ is the extended CMV matrix corresponding to the Verblunsky coefficients
\begin{equation}\label{e.correspondence}
\alpha_{2n+1} = 0 , \; \alpha_{2n} = \frac{\lambda_{2n}}{\lambda_{2n-1}} \bar c_n^{21}, \quad n \in \Z.
\end{equation}

\subsection{The Main Result}

We will consider the time-homogeneous case, that is, a sequence of coins that describe the update rule for any time transition. This sequence will only take two different values, and the order in which these two unitary $2 \times 2$ matrices occur is determined by an element of the Fibonacci subshift.

Let us recall how the latter is generated. Consider two symbols, $a$ and $b$. The Fibonacci substitution $S$ sends $a$ to $ab$ and $b$ to $a$. This substitution rule can be extended by concatenation to finite and one-sided infinite words over the alphabet $\{a,b\}$. There is a unique one-sided infinite word that is invariant under $S$, denote it by $u$. It is, in an obvious sense, the limit as $n \to \infty$ of the words $s_n = S^n(a)$. That is, $s_0 = a$, $s_1 =ab$, $s_2= aba$, etc., so that $u = abaababaabaab \ldots$. The Fibonacci subshift $\Omega$ is given by
$$
\Omega = \{ \omega \in \{ a,b \}^\Z : \text{every finite subword of $\omega$ occurs in } u \}.
$$
It is well known that there is $\omega^{(u)} \in \Omega$ such that
\begin{equation}\label{e.omegauchoice}
\omega^{(u)}\Big|_{n \ge 0} = u.
\end{equation}
In fact, there are exactly two possible choices for $\omega^{(u)} \in \Omega$ with this property, but we will choose and fix one here and work with this choice in the remainder of the paper.

Take $\theta_a , \theta_b \in (-\frac{\pi}{2}, \frac{\pi}{2})$ and consider the rotations
$$
C_a = \begin{pmatrix} \cos \theta_a & -\sin \theta_a \\ \sin \theta_a & \cos \theta_a \end{pmatrix}, \quad C_b = \begin{pmatrix} \cos \theta_b & -\sin \theta_b \\ \sin \theta_b & \cos \theta_b \end{pmatrix}.
$$
Given $\omega \in \Omega$, the associated sequence of coins $\{ C_{\omega,n} \}_{n \in  \Z}$ is given by $C_{\omega,n} = C_{\omega_n}$. The associated unitary operator will be denoted by $U_\omega$. Inspecting \eqref{e.correspondence} one sees that $U_\omega$ already has the form of an extended CMV matrix and we will therefore sometimes denote it by $\mathcal{E}_\omega$ to emphasize this fact. Using a strong approximation argument, it follows that the spectrum of $\mathcal{E}_\omega$ is independent of $\omega$ and hence may be denoted by $\Sigma$. (While this set does depend on the parameters $\theta_a , \theta_b$, our notation leaves this dependence implicit.)

Our goal is to establish estimates for the spreading of an initial state under the dynamics generated by $U_\omega$. For convenience, let us relabel the basis elements, ordered as in \eqref{e.orderedbasis}, and write them as $\{ e_m \}_{m \in \Z}$. We consider a non-zero finitely supported initial state $\psi \in \ell^2(\Z)$ and study the spreading in space of $U_\omega^n \psi$ as $|n| \to \infty$. Since the unitary operators $U_\omega$ we study are local in the sense that $\langle e_m , U_\omega e_k \rangle = 0$ if $|m-k| > 2$, it follows that for each $n \in \Z$, $U_\omega \psi$ is finitely supported as well. In particular, we can consider moments
$$
M_{\omega, \psi}(n,p) = \sum_{m \in \Z} (1 + |m|^p) |\langle e_m , U_\omega^n \psi \rangle|^2
$$
for $p > 0$, which are finite for all $n \in \Z$. Next we average in time and write
$$
\tilde M_{\omega, \psi}(N,p) = \frac{1}{N} \sum_{n=0}^{N-1} M_{\omega, \psi}(n,p).
$$
We ask at what power-law rate these quantities grow as $N \to \infty$ and hence set
$$
\tilde \beta^+_{\omega, \psi}(p) = \limsup_{N \to \infty} \frac{\log \tilde M_{\omega, \psi}(N,p)}{p \log N}, \quad \tilde \beta^-_{\omega, \psi}(p) = \liminf_{N \to \infty} \frac{\log \tilde M_{\omega, \psi}(N,p)}{p \log N}.
$$
These quantities are called transport exponents.

Our main result on the quantum walk in a time-homogeneous Fibonacci environment is the following:

\begin{theorem}\label{t.qwmain}
Define:
\begin{enumerate}

\item $C(z) = \max\{2+\sqrt{8+I(z)}, \rho^{-1} , \sigma^{-1}\}$, where $\rho = \sec\theta_a$ and $\sigma = \sec \theta_b$;

\item $\gamma_1(z) = \frac{\log \left(1+\frac{1}{4C(z)^2} \right)}{16 \log \phi}$, where $\phi$ is the golden mean;

\item $\gamma_2(z) = 4\log_2 K(z)$, where $K$ is a $z$-dependent constant described in the proof;

\item $\beta(z) = \frac{2\gamma_1 (z)}{\gamma_1 (z) + 2\gamma_2 (z)+1}$.

\end{enumerate}
Then, for all $\psi, \omega, p$ as above, we have
$$
\tilde \beta^\pm_{\omega, \psi}(p) \ge \max\big\{ \beta(z) : z\in \mathrm{supp}\ \mu_{U_\omega, \psi} \big\},
$$
where $\mu_{U_\omega, \psi}$ denotes the spectral measure associated with the unitary operator $U_\omega$ and the state $\psi$.
\end{theorem}

The lower bound provided by this theorem is an explicit function of the parameters of the system. While a weaker lower bound could be obtained by simply taking the minimum of $\beta(z)$ over the ($\omega$-independent) spectrum of $U_\omega$, the bound as stated captures the fact that the spreading estimates should indeed depend on where the spectral measure of the initial state is supported. This again is a reflection of the non-constancy of the so-called Fricke-Vogt invariant on the spectrum and is in marked contrast to the Schr\"odinger case, in which results of this type were first worked out, and where the invariant is in fact constant on the spectrum; see \cite{DKL} and references therein for related earlier work.

\subsection{The Associated Extended CMV Matrix and the Trace Map Formalism}

Consider the quantum walk model described in the previous subsection. Due to our choice of the angles $\theta_a$ and $\theta_b$, we have $c_{\omega,n}^{kk} = |c_{\omega,n}^{kk}|$, and therefore $\sigma^k_{\omega,n} = 0$ for every $n \in \Z$, $k \in \{1,2\}$. Thus, $\lambda_{\omega,n} = 1$ for every $n \in \Z$. By the CGMV connection, we are led to consider the extended CMV matrix $\mathcal{E}_\omega = U_\omega$, corresponding to the Verblunsky coefficients
\begin{equation}\label{e.fibcmv}
\alpha_{\omega,2n+1} = 0 , \; \alpha_{\omega,2n} = \bar c_{\omega,n}^{21} = \sin \theta_{\omega_n} \in (-1,1), \quad n \in \Z.
\end{equation}

In order to prove the main theorem, we will analyze the spectral measures associated with this extended CMV matrix. It is known that quantitative continuity properties of spectral measures imply spreading estimates of the form claimed in the previous subsection. These quantitative continuity properties in turn will be established through whole-line subordinacy theory. To describe what needs to be shown, let us recall the transfer matrices associated with a given sequence of Verblunsky coefficients. For $z \in \partial \D$, $\alpha \in \D$, and $\rho = (1 - |\alpha|^2)^{1/2}$, write
$$
T(z,\alpha) = \rho^{-1} \left( \begin{array}{cc} z & - \bar{\alpha} \\ - \alpha z & 1 \end{array} \right).
$$
Given a finite string $w = \alpha_1 \ldots \alpha_\ell$, with $\alpha_j \in \D$, let
$$
T(z,w) = T(\alpha_\ell,z) \cdots T(\alpha_1,z).
$$
Now, for $\omega \in \Omega$ and $n \in \Z_+$, we set
$$
T_n(z;\omega) = T(z,\alpha_{\omega,0} \ldots \alpha_{\omega,n-1}).
$$
Subordinacy theory considers sequences
\begin{equation}\label{e.tmequ1}
\begin{pmatrix} \xi_n \\ \zeta_n \end{pmatrix} = T_n(z;\omega) \begin{pmatrix} \xi_0 \\ \zeta_0 \end{pmatrix},
\end{equation}
where
\begin{equation}\label{e.tmequnorm1}
|\xi_0| = |\zeta_0| = 1.
\end{equation}
for energies $z$ in the spectrum of $\mathcal{E}_\omega$.

The following elementary observation will prove to be useful later.

\begin{lemma}\label{e.samesize}
Given $z \in \partial \D$ and a solution of \eqref{e.tmequ1} and \eqref{e.tmequnorm1}, we have $|\xi_n| = |\zeta_n|$ for every $n \ge 0$.
\end{lemma}

\begin{proof}
It suffices to show that in
$$
\begin{pmatrix} z & - \bar{\alpha} \\ - \alpha z & 1 \end{pmatrix} \begin{pmatrix} a \\ b \end{pmatrix} = \begin{pmatrix} az - b\bar \alpha \\ -a \alpha z + b \end{pmatrix},
$$
the fact that the absolute values of the entries are equal is preserved. That is, if $z \in \partial \D$ and $|a| = |b| \not= 0$, then $|az - b\bar \alpha| = |-a \alpha z + b|$.

Notice first that by linearity, we may assume that $a = 1$ and $b = \lambda \in \partial \D$ (first normalize and then multiply by the inverse of the first component). We have
\begin{align*}
|az - b\bar \alpha| & = |z - \lambda \bar \alpha| \\
& = |\bar \lambda z - \bar \alpha| \\
& = |\bar \lambda - \bar \alpha \bar z| \\
& = |\lambda - \alpha z| \\
& = |-a \alpha z + b|,
\end{align*}
which shows the desired identity.
\end{proof}

Since the Fibonacci sequence $u$ is invariant under $S$ and contains the symbols $a,b$, one can partition it not only by this pair of finite words, but also by $S^k(a), S^k(b)$ for every $k$. For the Verblunsky coefficient sequences $\{ \alpha_{\omega,n} \}_{n \in \Z}$, this suggests a similar partition. The inherent self-similarity gives rise to a trace-map formalism in which the traces of the transfer matrices over these basic building blocks play an essential role.

Some of the tools in the analysis had already been used in part one of this series \cite{DMY}. However, since the coefficients in \eqref{e.fibcmv} are different from the ones considered in \cite{DMY} due to the vanishing terms with odd index, we will need to account for that in our determination of the curve of initial conditions, which is the starting point of the approach put forward in \cite{DMY}.

Let us start by carrying out this calculation. We will later explain how this feeds into the theory. Denote the spectral parameter by $z$. Due to unitarity, all spectral measures will be supported by the unit circle $\partial \D$, and hence we will mainly be interested in the case $z \in \partial \D$. The three basic traces are the following:

\begin{align}
x_{-1}(z) & =  \Re(z) \sec\theta_b \label{e.trace1} \\
x_0(z) & = \Re(z) \sec\theta_a \label{e.trace2} \\
x_1(z) & = \Re(z^2)\sec\theta_a\sec\theta_b - \tan\theta_a\tan\theta_b.  \label{e.trace3}
\end{align}
Note in particular that all three basic traces are real-valued for every $z \in \partial \D$. This is crucial for the standard analysis of the trace map as a real dynamical system to be applicable in our setting.

The curve of initial conditions is
$$
(x_1(z), x_0(z), x_{-1}(z)), \quad z \in \partial \D.
$$

The associated Fricke-Vogt invariant is
$$
I(z) = x_1(z)^2 + x_0(z)^2 + x_{-1}(z)^2 - 2 x_1(z) x_0(z) x_{-1}(z) - 1, \quad z \in \partial \D.
$$
Equivalently,
\begin{align*}
I(z) & = \Re(z)^2(\sec^2\theta_a + \sec^2\theta_b) + (\Re(z^2) \sec\theta_a\sec\theta_b - \tan\theta_a\tan\theta_b)^2\\
& \qquad -2(\Re(z)^2 \sec^2\theta_a\sec^2\theta_b(\Re(z^2)-\sin\theta_a \sin\theta_b)) - 1,
\end{align*}

\subsection{Transfer Matrix Estimates}

Since for every $k \ge 1$, every $\omega \in \Omega$ can be partitioned into words of the form $S^k(a)$ and $S^k(b) = S^{k-1}(a)$, it is natural to express long products of transfer matrices in terms of these basic building blocks. Notice however that due to \eqref{e.correspondence} we have to insert a zero coefficient after each symbol in $S^k(a)$ to obtain the appropriate associated building block $\tilde s_k$ of Verblunsky coefficients, which has length $2F_k$, where $F_k$ is the $k$-th Fibonacci number. That is, we have
\begin{align*}
\tilde s_0 & = (\sin \theta_a) \, 0 , \\
\tilde s_1 & = (\sin \theta_a) \, 0 \, (\sin \theta_b) \, 0, \\
\tilde s_2 & = (\sin \theta_a) \, 0 \, (\sin \theta_b) \, 0 \, (\sin \theta_a) \, 0, \\
\tilde s_3 & = (\sin \theta_a) \, 0 \, (\sin \theta_b) \, 0 \, (\sin \theta_a) \, 0 \, (\sin \theta_a) \, 0 \, (\sin \theta_b) \, 0, \\
& \; \; \vdots
\end{align*}

Let
$$
M_k(z) = T(z,\tilde s_k)
$$
and
$$
x_k(z) = \frac12 \mathrm{Tr} \, M_k(z).
$$
The recursion $\tilde s_{k+1} = \tilde s_k \tilde s_{k-1}$ for the words gives rise to the recursion
\begin{equation}\label{e.matrec}
M_{k+1}(z) = M_{k-1}(z) M_k(z)
\end{equation}
for the matrices and the recursion
\begin{equation}\label{e.tracerec}
x_{k+1}(z) = 2 x_k(z) x_{k-1}(z) - x_{k-2}(z)
\end{equation}
for the half-traces. Note that \eqref{e.tracerec} is invertible and determines values for $x_k(z)$ for $k \le 0$ as well. It is easy to check from the definitions that for $k = 1,0,-1$, we indeed get \eqref{e.trace1}--\eqref{e.trace3} for $x_k(z)$.

For $\xi : \Z_{\ge 0} \to \C$ and $L \ge 1$, write
$$
\norm{\xi}_L^2 = \sum_{n = 0}^{\lfloor L \rfloor} |\xi_n|^2 + (L - \lfloor L \rfloor) |\xi_{\lfloor L \rfloor + 1}|^2.
$$
Our goal is to prove power-law upper and lower bounds for $\norm{\xi}_L$, where $\xi$ solves \eqref{e.tmequ1} subject to \eqref{e.tmequnorm1}, uniformly in $\omega \in \Omega$. As mentioned above, one may then derive continuity properties of spectral measures from such estimates, and this in turn implies the desired lower bound for the transport exponents. Again, the second and third step are model-independent, so the model-dependent part of the analysis happens in the first step, where the sequences $\xi$ generated by \eqref{e.tmequ1} and \eqref{e.tmequnorm1} are estimated. These estimates will be proved in the present subsection. The other steps have been addressed (in a general context) in separate publications \cite{DFV, MO} and the relevant results from those papers are summarized in the appendix.

\begin{lemma}
For $z \in \Sigma$, we have $\sup |x_k (z)| \le C(z)$ with $C(z)$ as in Theorem~\ref{t.qwmain}.
\end{lemma}

\begin{proof}
Since the traces $x_k$ obey the recursion \eqref{e.tracerec}, one can mimic the proof of the analogous estimate in \cite[Proposition~12.8.6]{S2}. This yields $\sup |x_k (z)| \le \max\{ 2+\sqrt{I(z) + 8}, |x_{-1}(z)|, |x_0 (z)| \}$, which implies the required estimate.
\end{proof}

Define $T$ to be the transfer matrix corresponding to a Verblunsky coefficient of zero, and $A$, $B$ the transfer matrices corresponding to $\theta_a$ and $\theta_b$. That is,
\begin{align*}
A(z) & = \sec\theta_a\begin{pmatrix} z & -\sin \theta_a \\ - z \sin \theta_a & 1 \end{pmatrix}, \\
B(z) & = \sec\theta_b\begin{pmatrix} z & -\sin \theta_b \\ - z \sin \theta_b & 1 \end{pmatrix}, \\
T(z) & = \begin{pmatrix} z & 0 \\ 0 & 1 \end{pmatrix}.
\end{align*}

\begin{lemma}\label{l.nests}
We have the following estimates:
\begin{enumerate}

\item $\norm{T(z) A(z)}$ is {\rm (}for all $z${\rm )} equal to $\sec\theta_a (1+|\sin\theta_a|)^2$.

\item $\norm{T(z) B(z) T(z) A(z)} \le 12(\sec\theta_a \sec\theta_b)^{3/2}$.

\item $\norm{T(z) A(z) T(z) B(z) T(z) A(z)} \le 48(\sec\theta_a)^{5/2} (\sec\theta_b)^{3/2}$.

\end{enumerate}
\end{lemma}

\begin{proof}
These estimates were computed with a symbolic algebra package.
\end{proof}

\begin{prop}\label{p.tmubounds}
For $z \in \Sigma$ and every sequence $\xi$ generated by \eqref{e.tmequ1} and \eqref{e.tmequnorm1}, we have
$$
\norm{\xi}_L \le C_2(z) L^{2\gamma_2 (z)+1}
$$
for $L \ge 1$, uniformly in $\omega$. The constant $C_2(z)$ is irrelevant to the long-term spreading behavior, so it is not calculated here. The treatment in \cite{MO} contains an expression for it.
\end{prop}

\begin{proof}
Let $\omega \in \Omega$ and $z \in \Sigma$, and consider a sequence $\xi$ generated by \eqref{e.tmequ1} and \eqref{e.tmequnorm1}.

Clearly, a power law upper bound on $\norm{T_n(z;\omega)}$ implies one for $\norm{\xi}_L$. To be precise (taking $L$ to be an integer for notational simplicity), suppose that $\norm{T_n(z;\omega)} \le \bar C_2(z) n^{\gamma(z)}$ for $n \ge 1$. Then,
$$
\norm{\xi}_L^2 = \sum_{n=0}^L |\xi_n|^2 \le 1 + \bar C_2(z)^2 \sum_{n=1}^L n^{2\gamma(z)} \le \tilde C_2(z)^2 L^{2\gamma(z)+1}
$$
for $L \ge 1$. Thus, $\norm{\xi}_L \le \tilde C_2(z) L^{\gamma(z) + 1/2}$; and hence our goal is to prove a power law bound on $\norm{T_n(z;\omega)}$.

The hypotheses of \cite[Theorem~3]{MO} (which adapts \cite{DKL} and \cite{IRT92} from OPRL to OPUC) are satisfied by the transfer matrices at hand for $\omega = \omega^{(u)}$ (with $\omega^{(u)}$ from \eqref{e.omegauchoice}), so it can be used to prove power law bounds on $T_n(z;\omega^{(u)})$. The bounds proved in Lemma~\ref{l.nests} make it easy to verify the hypotheses of \cite[Theorem~3]{MO}, and we may conclude:

\begin{prop}
For all $z \in \Sigma$, there exist $\gamma_2(z)$ and $\bar C_2(z)$ independent of $n$ such that
$$
\norm{T_n(z;\omega^{(u)})} \le \bar C_2(z) n^{\gamma_2(z)}
$$
for $n \ge 1$. The exponent is given by $\gamma_2(z) = 4 \log_2 K(z)$, where
\begin{align*}
K(z) & = \max \Big( 8 \max(1, \sup_k |x_k(z)|), 4 \norm{T(z) A(z)}, \\
& \qquad 4 \norm{T(z) B(z) T(z) A(z)}, 4 \norm{T(z) A(z) T(z) B(z) T(z) A(z)} \Big) \\
& \qquad \times (4 + 4 \max (1, \sup_k |x_k(z)|)).
\end{align*}
\end{prop}

Thus for $\omega = \omega^{(u)}$, we obtain $\norm{\xi}_L \le \tilde C_2(z) L^{\gamma_2(z) + 1/2}$. It remains to establish the claimed estimate for other choices of $\omega \in \Omega$.

Lemma~5.2 of \cite{DamanikLenz1999} does this for Schr\"odinger operators. Its proof is model-independent and applies here. It works by splitting the word $w$ over which $M$ runs into a prefix and a suffix of the zero-phase word, the latter of which is then reduced by a suitable constant-length portion to yield the reversal of a prefix. The above estimate is applied to the prefix. Then, the lemma below is used to apply the zero-phase estimate to the reduced suffix, obtaining the claimed exponent $2(\gamma_2(z) + 1/2)$.

\begin{lemma}\label{l.wreversal}
Given a string $w = \alpha_1 \ldots \alpha_\ell$ of Verblunsky coefficients, we denote its reversal $\alpha_\ell \ldots \alpha_1$ by $w^R$. For all $z \in \partial \D$, we have $\norm{T(z,w^R)} = \norm{T(z,w)}$.
\end{lemma}

\begin{proof}
Because $|z|=1$, $\norm{T(z,w)} = \norm{\frac{1}{z^{|w|/2}} T(z,w)}$, and hence we may pass to unimodular matrices. Using that $\mathrm{SU}(1,1)$ is unitarily conjugate to $\mathrm{SL}(2,\R)$ (see \cite[Section~10.4]{S2}) and the proof of the analogous lemma in \cite{DamanikLenz1999} works for all $\mathrm{SL}(2,\R)$ matrices, the claim follows.
\end{proof}
This completes the proof of Proposition~\ref{p.tmubounds}.
\end{proof}

\begin{prop}\label{p.tmlbounds}
For $z \in \Sigma$ and every sequence $\xi$ generated by \eqref{e.tmequ1} and \eqref{e.tmequnorm1}, we have
$$
\norm{\xi}_L \ge C_1(z) L^{\gamma_1 (z)}
$$
for $L \ge 1$, uniformly in $\omega$. As before, the constant $C_1(z)$ is irrelevant to the long-term spreading behavior, so it is not calculated here, and the treatment in \cite{MO} contains an expression for it.
\end{prop}

\begin{proof}
Due to Lemma~\ref{e.samesize} it suffices to prove lower bounds for (the $\|\cdot\|_L$-norm associated with) the vector-valued sequence
$$
\begin{pmatrix} \xi_n \\ \zeta_n \end{pmatrix} = T_n(z;\omega) \begin{pmatrix} \xi_0 \\ \zeta_0 \end{pmatrix}
$$
in \eqref{e.tmequ1}. Such bounds can be obtained via an inductive scheme based on the Gordon two-block method developed in \cite{DKL}. The necessary two-block structures in elements $\omega$ of $\Omega$ have been exhibited in \cite{DKL, DamanikLenz1999a}. The overall setup of the proof is the same in the case at hand, the model-dependence solely resides in the initial conditions. Bearing Lemma~\ref{l.nests} in mind, this approach yields the desired lower bound with
$$
\gamma_1(z) = \frac{\log \left( 1+\frac{1}{4C(z)^2} \right)}{16\log \phi}.
$$
We refer the reader to the discussion in \cite{MO} for further details; compare especially Theorem~4 of that paper.
\end{proof}

\subsection{Conclusion of the Proof of Theorem~\ref{t.qwmain}}

Let $V \subset \Sigma \subset \T$ be the topological support of $\mu_{U,\psi}$. The above power-law bounds show, by Proposition~\ref{p.dklest}, that for $z \in V$, the local scaling exponent of $\mu_{U,\psi}$ at $z$ is bounded from below by
$$
\beta(z) = \frac{2\gamma_1 (z)}{\gamma_1 (z) + 2\gamma_2 (z)+1}.
$$
Since $V$ is compact and $\beta(\cdot)$ is continuous, $\mu_{U,\psi}$ has Hausdorff dimension at least $\max \{ \beta(z) : z \in V \}$. By Proposition~\ref{p.gclest}, we obtain
$$
\tilde \beta^\pm_{\omega, \psi}(p) \ge \max\big\{ \beta(z) : z\in \mathrm{supp}\ \mu_{U_\omega, \psi} \big\}
$$
and Theorem~\ref{t.qwmain} follows. \hfill\qedsymbol

\section{Applications to $1D$ Nearest Neighbor Ising Ferromagnets}\label{sec:Ising}

In this section we establish a connection between a class of Ising models on the one-dimensional lattice, immersed in a complex magnetic field, and CMV matrices. Specifically, we consider nearest-neighbor ferromagnetic models. We then use this connection and the results on CMV matrices with Fibonacci Verblunsky coefficients that were developed in \cite{DMY} to give a complete description of Lee-Yang zeros in the thermodynamic limit, in the complex fugacity variable, when the nearest-neighbor interaction is modulated by the Fibonacci substitution sequence.

\subsection{Introduction}\label{sec:Ising-intro}

For a historical account of Ising models,  see, for example, \cite{Brush1967}, and for technical details, see \cite{Baxter1982}. Here we briefly describe the model and introduce only the necessary notions.

Let $\Lambda_N := \set{\pm 1}^N$, with $N\in\N$. For a \textit{configuration} $\sigma = (\sigma_0,\dots,\sigma_N)\in \Lambda_N$, we define the energy of $\sigma$, $E(\sigma)$, by
\begin{align}\label{eq:ising-hamiltonian}
 E(\sigma) := -\frac{1}{k_B\tau}\sum_{i = 1}^N(J_i\sigma_i\sigma_{i+1} + H\sigma_i),
\end{align}
with $\sigma_{N+1} = \sigma_0$, and $\set{J_i}\subset\R$. Here $\tau\in(0,\infty)$ is the temperature, and $k_B>0$ is the so-called Boltzmann constant (included here for historical reasons, but is often factored into $\tau$). The sequence $\set{J_i}$ defines a nearest-neighbor interaction; that is, each $J_i$ gives the strength of interaction between neighboring \textit{spins} $\sigma_i$ and $\sigma_{i+1}$. The constant $H$ is the \textit{external magnetic field}, which takes values in $\C$. The tuple $(\Lambda_N, \set{J_i}, \tau, H, E)$ then defines a (one-dimensional) \textit{Ising model} on the finite lattice of length $N$, of temperature $\tau$, immersed in the magnetic field of strength $H$ and with energy $E$ (we are not giving here the most general definition of an Ising model). As has already been mentioned, here we are concerned with the \textit{ferromagnetic} model (that is, each $J_i > 0$). Moreover, we will concentrate on the case where the sequence $\set{J_i}$ is modulated by a sequence from the
one-sided subshift generated by Fibonacci substitution. That is, let $u$ denote as above the sequence invariant under the Fibonacci substitution and let $\Omega$ be the associated one-sided subshift consisting of those one-sided infinite sequences that locally look like $u$.

Let $p: \set{a,b}\rightarrow\R_{>0}$. Then for a given $\omega\in\Omega$, we take $J_i = p(\omega_i)$. The case where $\omega = u$ has been considered in a number of papers (see, for example, \cite{Baake1995, Barata2001} and references therein), and the more general case (with variable $\omega$) was recently considered in \cite{Y3}.

Let us now describe the main object of interest - the so-called \textit{partition function}. We begin by noting that while there are infinitely many ways to define a probability measure on $\Lambda_N$, it turns out, from the statistical physics point of view, that the most natural one is
\begin{align*}
 \mathbb{P}_N(\sigma) = \frac{e^{-E(\sigma)}}{\sum_{\sigma\in\Lambda_N}e^{-E(\sigma)}}.
\end{align*}
The measure $\mathbb{P}_N$ is a so-called \textit{Gibbs state}, which is a probability measure that maximizes the entropy of the system given by
\begin{align*}
 -\sum_{\sigma\in\Lambda_N}\mathbb{P}_N(\sigma)\log\mathbb{P}_N(\sigma).
\end{align*}
According to the Boltzmann hypothesis, if a thermodynamic system is in equilibrium, then it is in a Gibbs state (or, more precisely, fluctuates around a Gibbs state). Obviously $\mathbb{P}_N$ is completely determined by the \textit{partition function}
\begin{align*}
 Z^{(N)} := \sum_{\sigma \in\Lambda_N}e^{-E(\sigma)}
\end{align*}
(considered as a function of relevant parameters -- say $\tau$ or $H$).

It turns out that $Z^{(N)}$ encodes all the thermodynamic information about the system. Most notably, zeros of $Z^{(N)}$ (or of its derivatives)---equivalently, singularities of $\mathbb{P}_N$---indicate critical behavior (or \textit{phase transitions}), which is one of the central research directions in modern statistical mechanics. It is the zeros of $Z^{(N)}$ that we study here.

For an arbitrary choice of the sequence $\set{J_i}\subset\R_{>0}$, and $\tau\in(0,\infty)$ fixed, let us define the following variables
\begin{align*}
 \beta_i := \exp \frac{2J_i}{k_B\tau}\hspace{1cm}\text{ and }\hspace{1cm}h:=\exp\frac{2H}{k_B\tau},
\end{align*}
and consider the partition function $Z^{(N)}$ as a function of $h$ ($h$ is varied by varying $H$). It turns out that zeros of $Z^{(N)}(h)$ (in the variable $h$) lie on the unit circle (this is a consequence of the famous Lee-Yang theorem in statistical mechanics \cite{Lee1952}). Moreover, the distribution of these zeros (both topological and statistical) in the limit $N\rightarrow\infty$ (the so-called \textit{thermodynamic limit}) provides information about critical behavior of the model as the lattice size tends to infinity. What makes an analysis of zeros in the thermodynamic limit possible (at least computationally) is the applicability of the transfer matrix formalism:
\begin{align}\label{eq:pfunc-tmat}
 Z^{(N)}(h) = \mathrm{Tr}\prod_{i = N}^0 M_i(h),\hspace{2mm}\text{ where }\hspace{2mm}
 M_i(h) = (\beta_i h)^{-1/2}
 \begin{pmatrix}
        \beta_i h & \sqrt{h}\\
        \sqrt{h} & \beta_i
  \end{pmatrix}
\end{align}
(see \cite{Baxter1982}); observe also that neither $\beta_i$ nor $h$ is ever zero for any $i$. Notice that the zeros of $Z^{(N)}(h)$ coincide with the zeros of
\begin{align}\label{eq:pfunc-tmat-mod}
 \widetilde{Z}^{(N)}(h) = \mathrm{Tr}\prod_{i = N}^0\widetilde{M}_i(h),\hspace{2mm}\text{ where }\hspace{2mm}\widetilde{M}_i(h) =
 \begin{pmatrix}
  h & \frac{\sqrt{h}}{\overline{\beta_i}}\\
  \frac{\sqrt{h}}{\beta_i} & 1
 \end{pmatrix}.
\end{align}
(after factoring out $\beta_i$ from the matrix in the definition of $M_i(h)$, and dropping the factor $(\beta_i h)^{-1/2}\beta_i$ since $\beta_i$ and $h$ are nonzero). The complex conjugation of $\beta_i$ in $\widetilde{M}_i$ above may seem redundant (since $\beta_i$ is real), but in the next section we shall extend $\widetilde{Z}^{(N)}$ to complex-valued $\beta_i$.

\subsection{Relation to CMV Matrices}\label{subsec:Ising-CMV}

Denote the open unit disc in $\C$ by $\mathbb{D}$, its closure by $\overline{\mathbb{D}}$, and the complement of its closure by $\mathbb{G}$. Denote by $\mathbb{D}_\circ$ the unit disc $\mathbb{D}$ with the origin removed. Denote by $S^\infty$ the infinite (countable) product of a set $S$ with itself.

For a sequence $u\in \mathbb{G}^\infty$, let $\Theta u$ denote its inversion:
\begin{align*}
 (\Theta u)_n = \frac{1}{u_n}.
\end{align*}
Obviously $\Theta: \mathbb{G}^\infty\rightarrow\mathbb{D}_\circ^\infty$ is a bijection.

For a given sequence $J = \set{J_i}$ real and positive, the partition function $Z^{(N)}(h)$ was defined above, and is completely determined by the corresponding sequence $\beta = \set{\beta_i}$ in \eqref{eq:pfunc-tmat}, and the zeros of $Z^{(N)}(h)$ coincide with those of $\widetilde{Z}^{(N)}(h)$ from \eqref{eq:pfunc-tmat-mod}. By \eqref{eq:pfunc-tmat-mod}, $\widetilde{Z}^{(N)}(h)$ can be easily extended to the case where the sequence $\set{J_i}$ is complex-valued. In case $\Re(J_i)>0$ for each $i$, the corresponding sequence $\beta$ lies in $\mathbb{G}^\infty$ (the reader who is familiar with the Lee-Yang theorem will certainly notice that this condition is in line with the assumptions on the models that were originally considered in \cite{Lee1952}; in fact, $\Re(J_i) > 0$ insures that the Lee-Yang zeros lie on the unit circle by the original Lee-Yang theory, and it also follows from our computations below). Hence $\Theta\beta\in \mathbb{D}_\circ^\infty$. Thinking of $\Theta\beta$ as a sequence of
Verblunsky coefficients, we may associate to each partition function (which, again, is completely determined by $\beta$, or, equivalently, by the choice of the sequence $\set{J_i}$) a CMV matrix with Verblunsky coefficients $\Theta\beta$ in a bijective way.

On the other hand, to each CMV matrix there is associated a sequence of transfer matrices given in \eqref{e.tmbasic}. As above, for a given CMV operator with Verblunsky coefficients $\set{\alpha_i}$, denote the $n$-step transfer matrix by $T_n$:
\begin{align}\label{eq:CMV-transfer}
 T_n(w) := \prod_{j = n-1}^0(1-\abs{\alpha_j}^2)^{-1/2}
 \begin{pmatrix}
  z & -\overline{\alpha_j}\\
  -\alpha_j w & 1
 \end{pmatrix}.
\end{align}
We have
\begin{prop}\label{prop:Ising-zeros}
 For a given $\beta\in\mathbb{G}^\infty$, the zeros of $\widetilde{Z}^{(N)}(h)$ coincide with those of $\mathrm{Tr}(T_{N+1}(h))$ with Verblunsky coefficients $\Theta\beta$.
\end{prop}

In what follows, we shall use the notation $\Delta_{N+1}(h) := \mathrm{Tr}(T_{N+1}(h))$. The polynomial $\Delta$ is called the \textit{discriminant}.

\begin{remark}
 The discriminant above is defined in a way that is slightly different from how it is defined in equation (11.1.2) in \cite{S2}, where it is defined on $\mathbb{D}_\circ$ as $z^{\frac{N+1}{2}}\mathrm{Tr}T_{N+1}(z)$ ($N+1$ is taken to be even for simplicity, but this is not necessary). On the other hand, in what follows, we shall consider the zeros of $\Delta_{N+1}$, and as Theorem 11.1.1 in \cite{S2} states, the zeros lie on $\partial\mathbb{D}$ (and are simple). For this reason we can drop the factor $z^{\frac{N+1}{2}}$.
\end{remark}

\begin{proof}
 Define a matrix $D := \left(\begin{smallmatrix} -1 & 0 \\ 0 & h^{-1/2}\end{smallmatrix}\right)$. For a given $\beta\in\mathbb{G}^\infty$, let $\set{\alpha_i} = \Theta\beta$. Then we get, for each $i = 0, 1,\dots$,
 \begin{align*}
  D^{-1}\widetilde{M}_i(h)D = \begin{pmatrix}
   h & -\overline{\alpha_j}\\
   -\alpha_jh & 1
  \end{pmatrix}.
 \end{align*}
The result follows after noting that in \eqref{eq:CMV-transfer}, $(1 - \abs{\alpha_j}^2)^{-1/2} > 0$ for each $i = 0, 1,\dots$ (so contribution of zeros comes only from the matrices in \eqref{eq:CMV-transfer}), and that $D$ depends only on $h$ and not on $\alpha = \Theta\beta$.
\end{proof}

We can now exploit Proposition \ref{prop:Ising-zeros} to prove the main result of this section:
\begin{theorem}\label{thm:Ising-CMV}
 If $p: \set{a,b}\rightarrow \mathbb{G}$, then for any $\omega\in\Omega$, if $\beta = \set{\beta_i = p(\omega_i)}\subset\mathbb{G}^\infty$, then the zeros of the corresponding $\widetilde{Z}^{(N)}$ accumulate in the limit $N\rightarrow\infty$ on the essential spectrum of the CMV matrix with Verblunsky coefficients given by $\Theta\beta$.
\end{theorem}

An immediate consequence is the following corollary (see \cite[Section 12.8]{S2} where it is proved that the essential spectrum of the CMV matrix corresponding to Verblunsky coefficients that follow a sequence $\omega \in \Omega$ is independent of $\omega$).

\begin{coro}\label{coro:Ising-CMV}
 With the hypothesis of Theorem \ref{thm:Ising-CMV}, the topological limit behavior of the zeros of $\widetilde{Z}^{(N)}$ is independent of $\omega\in\Omega$.
\end{coro}

\begin{remarks}
 Before we continue, a few remarks are in order:
 \begin{enumerate}
  \item By accumulation in the limit $N\rightarrow\infty$, we mean that there exists a sequence $\set{N_i}\subset\N$, $N_i\uparrow\infty$, such that the zeros of $\widetilde{Z}^{(N_i)}$ converge in the Hausdorff metric.

  \item In the case $\omega = u$, the problem was studied numerically in \cite{Baake1995}, heuristically in \cite{Barata2001}, and rigorously in \cite{Y3}; however, the dependence (or lack thereof) on $\omega$ remained a mystery.

  \item A connection between the Lee-Yang zeros and spectra of linear operators has been postulated by a few authors. Here we (constructively) establish such a connection for a specific class of models.
 \end{enumerate}
\end{remarks}

\begin{proof}[Proof of Theorem \ref{thm:Ising-CMV}]
 Let $\omega = u$, and let $\mathcal{C}$ be the CMV matrix with Verblunsky coefficients given by $\Theta\beta$, and $\mathcal{C}_{F_k}$ denotes the $F_k$th periodic approximation of $\mathcal{C}$, as described in \cite[Section 3.2]{DMY}. Then a combination of Proposition 3.3 and Theorem 2.18 from \cite{DMY} gives:
 \begin{align*}
  \sigma_\mathrm{ess} (\mathcal{C}_{F_k})\rightarrow \sigma_\mathrm{ess} (\mathcal{C})\hspace{2mm}\text{ in the Hausdorff metric as}\hspace{2mm} k\rightarrow\infty.
 \end{align*}
 Moreover, Theorem 3.2 from \cite{DMY} gives
 \begin{align}\label{eq:periodic-spec}
  \sigma_\mathrm{ess} (\mathcal{C}_{F_k}) = \Delta^{-1}[-1,1]
 \end{align}
(with $\Delta$ being the discriminant, as given in Definition 3.1 in \cite{DMY}), which is composed of $F_k$ arcs on $S^1$ whose interiors are disjoint. On the other hand, the zeros of the discriminant $\Delta$, which are the zeros of $\mathrm{Tr}(T_{F_k})$, the trace of the $F_k$-step transfer matrix, lie inside of these arcs (one zero per arc). By Proposition \ref{prop:Ising-zeros}, it follows that when $\omega = u$, the zeros of $\widetilde{Z}^{(N)}$ accumulate on $\sigma_\mathrm{ess} (\mathcal{C})$; more precisely, if $\hdist$ denotes the Hausdorff metric, then we have
\begin{align*}
 \lim_{k\rightarrow\infty}\hdist(\mathrm{zeros}(\widetilde{Z}^{(F_k)}), \sigma_\mathrm{ess} (\mathcal{C})) = 0.
\end{align*}

Now take an arbitrary $\omega \in \Omega$, and let $\mathcal{C}_\omega$ be the CMV matrix with Verblunsky coefficients $\Theta\beta$. Again, let $\mathcal{C}_{\omega, F_k}$ be the $F_k$-periodic approximation. The proof of the following lemma will complete the proof of the theorem (namely, it demonstrates independence of $\omega$):
\begin{lemma}\label{lem:repetition}
 There exists an increasing subsequence $\set{{F}_{k_i}}\subset\set{F_k}$, such that $\sigma_\mathrm{ess} (\mathcal{C}_{\omega,{F}_{k_i}}) = \sigma_\mathrm{ess} (\mathcal{C}_{{F}_{k_i}})$ for all $i = 1, 2,\dots$.
 \end{lemma}

\begin{proof}[Proof of Lemma \ref{lem:repetition}]
The matrix $\mathcal{C}_{\omega,F_k}$ is by definition one-sided, but it can be extended to a bi-infinite matrix in a natural way (see \cite[Chapter 11]{S2}); denote this extension by $\mathcal{E}_{\omega, F_k}$. Similarly, extend $\mathcal{C}_{F_k}$, and denote the extension by $\mathcal{E}_{F_k}$. We have $\sigma(\mathcal{E}_{\omega,F_k}) = \sigma_\mathrm{ess} (\mathcal{C}_{\omega,F_k})$ and $\sigma(\mathcal{E}_{F_k}) = \sigma_\mathrm{ess} (\mathcal{C}_{F_k})$. We will show that for some increasing subsequence $\set{{F}_{k_i}}\subset\set{F_k}$, $\mathcal{E}_{\omega,{F}_{k_i}}$ is unitarily equivalent to $\mathcal{E}_{{F}_{k_i}}$.

\begin{definition}
We say that a finite subword $w$ of the Fibonacci sequence $u$ is \textup{repeatable} if the bi-infinite periodic sequence with the unit cell $w$ is a finite shift of the bi-infinite periodic sequence with the unit cell $u_0u_1\cdots u_{\abs{w}-1}$, where $\abs{w}$ denotes the length of $w$. A finite subword is said to be \textup{nonrepeatable} if it is not repeatable.
\end{definition}

(Certainly not every subword of $u$ is repeatable -- take, for example, $aa$.)

\begin{lemma}\label{lem:repeatability}
 For any subword $w$ of $u$ with $\abs{w} = F_k$, with $F_k \geq 2$, we have the following. Write $w = w_1w_2$, with $\abs{w_1} = F_{k-1}$. Then either $w$ is repeatable, or $w_1$ is repeatable.
\end{lemma}

\begin{remark}
 We shall demonstrate below that for each $F_k\geq 1$, there is a unique word of length $F_k$ that is not repeatable.
\end{remark}

\begin{proof}[Proof of Lemma \ref{lem:repeatability}]
Let us recall quickly how the sequence $u$ is constructed. Let $S(a) = ab$ and $S(b) = a$, then iterate $S$ starting with $a$:
\begin{align*}
 S: a\mapsto ab\mapsto aba\mapsto abaab\mapsto apaababa\mapsto abaababaabaab\mapsto\cdots.
\end{align*}
Notice that at each step the word is a concatenation of the previous two words. Let $W_i$ be the word at level $i$; that is, $W_i = S^{i - 1}(a)$, with $W_1 = a$. Then $W_k = W_{k-1}W_{k-2}$, and $\abs{W_k} = F_k$.
\begin{claim}\label{claim:arb1}
The following statements hold.
\begin{enumerate}
 \item For each $k\geq 2$, the word $W_kW_k$ contains exactly $F_k$ distinct subwords of length $F_k$, all of which are repeatable.
 \item For each $k$, there is precisely one subword of $u$ of length $F_k$ which is not repeatable, which is characterized as follows. If $S^{k+1}(a) = W_{k+1}W_k$, then the nonrepeatable subword of length $F_k$ is the subword $w$ of $W_{k+1}W_k$ such that $W_{k+1}W_k = vwx$ with $\abs{x} = 1$.
\end{enumerate}
\end{claim}
We should remark that the statement (1) of Claim \ref{claim:arb1} is known (and statement (2) is a trivial consequence of (1)), and appears implicitly in a number of works (see, for example, \cite{DamanikLenz2002, DamanikLenz2003}). The proof is quite straightforward by induction; for details, see, for example, \cite[Chapter 2]{Lothair2002} and \cite[Chapter 3]{Smyth2003}.

To complete the proof of the lemma, assume that $w$ is a subword of $u$ of length $F_k\geq 2$ which is not repeatable. Write $w = w_1w_2$, where $\abs{w_1} = F_{k-1}$.  Then by characterization of nonrepeatable words in Claim \ref{claim:arb1}, $w_1$ is a subword of $W_{k-1}W_{k-1}$, which is repeatable.
\end{proof}

The following result is easily deduced from Lemma~\ref{lem:repeatability}; one only needs to notice that by definition, for any $\omega \in \Omega$, any finite subword of $\omega$ is a subword of $u$.

\begin{lemma}\label{lem:hull-repeatable}
 For each $\omega\in\Omega$, there exists an increasing subsequence $\set{{F}_{k_i}}\subset\set{F_k}$, such that each subword $\omega_{{F}_{k_i}} := \omega_0\omega_1\cdots\omega_{{F}_{k_i}}$ of $\omega$ is a subword of $u$, and is repeatable.
\end{lemma}

To finish the proof of Lemma \ref{lem:repetition}, notice that by Lemma \ref{lem:hull-repeatable}, for any $\omega\in\Omega$, there exists an increasing sequence $\set{{F}_{k_i}}\subset\set{F_k}$, such that $\mathcal{C}_{\omega,{F}_{k_i}}$ is unitarily equivalent to $\mathcal{C}_{{F}_{k_i}}$ via a left shift by the appropriate number of places.
\end{proof}
Application of Lemma \ref{lem:repeatability} now completes the proof of Theorem \ref{thm:Ising-CMV}.
\end{proof}

\subsection{The Density of Zeros Measure}\label{sec:doz}

We have so far discussed the topological distribution of the Lee-Yang zeros in the thermodynamic limit. Another natural question concerns the measure-theoretic distribution of the zeros. More precisely: denote by $\mathcal{Z}^{(N)}$ the set of zeros of $Z^{(N)}$ and by $\mathcal{Z}$ the zeros in the thermodynamic limit, and let $\mu_n$ denote the counting probability measure on $\mathcal{Z}^{(N)}$; \textit{does there exist a measure $\mu$ on the unit circle, supported on $\mathcal{Z}$, which is a (weak) limit of the measures $\set{\mu_n}$ (or a subsequence thereof)? If so, how does $\mu$ relate to the corresponding CMV matrix?}

Now that we have established a connection between Ising models and CMV matrices, the above questions have straightforward answers. Indeed, we have

\begin{theorem}\label{thm:IDS-Ising}
 The measures $\mu_n$ converge weakly to the density of states measure of the corresponding CMV matrix, supported on the essential spectrum thereof.
\end{theorem}

In fact, we shall prove a more general result for ergodic Ising models, of which the Fibonacci quasi-periodic model above is a special case.

Indeed, suppose $(\Omega, g, \mu, T)$ is a stochastic system; that is, $\Omega$ is a topological space with $\mu$ a positive Borel probability measure, $g:\Omega\rightarrow\C$ a bounded $\mu$-measurable function, and $T: \Omega\rightarrow\Omega$ is ergodic with respect to $\mu$ (let us also assume that $T$ is invertible; for if not, there is an invertible system measure-theoretically conjugate to $(\Omega, \mu, T)$).

\begin{theorem}\label{thm:IDS-Ising-general}
 Let $(\Omega, g, \mu, T)$ be a stochastic system. Suppose further that $g:\Omega\rightarrow\mathbb{G}$, and set $\beta_j(\omega) = g(T^j(\omega))$. Let $d\nu_n(\omega)$ be the counting probability measure on the zeros of $Z^{(n)}$. If $\mathbb{E}(\log\beta_j(\omega)) < \infty$, then there exists a probability measure $d\nu$ (independent of $\omega$) on $\partial \mathbb{D}$ such that for almost every $\omega\in\Omega$, $d\nu_n(\omega)\rightarrow d\nu$ weakly. Moreover, $d\nu$ is precisely the density of states measure (as defined in Theorem 10.5.21 of \cite{S2}) for the CMV matrix with Verblunsky coefficients given by $\Theta\beta_j(\omega)$.
\end{theorem}

\begin{proof}
 In what follows, we consider $\Delta_n$ restricted to $\partial\mathbb{D}$. Set $B_n = \Delta_n^{-1}[-2,2]$. Theorem 11.1.1 from \cite{S2} guarantees that $B_n$ consists of $n+1$ compact arcs with nonempty interior in $\partial \mathbb{D}$, that we call $\sigma_i^{(n)}$, $i = 1, \dots, n + 1$, and the roots $z_1^{(n)}, z_2^{(n)}, \dots, z_{n+1}^{(n)}$ of $\Delta_n$ lie on $\partial \mathbb{D}$ with $z_i^{(n)}$ in the interior of $\sigma_i^{(n)}$; moreover, the arcs $\sigma_i^{(n)}, \sigma_j^{(n)}$, $i\neq j$, may intersect only at the endpoints. Now let $z_i^{(n), r}$ denote the right end-point of $\sigma_i^{(n)}$. Define
 \begin{align*}
  \gamma_n := (-1)^{n+1}\prod_{i = 1}^{n+1} z_i^{(n), r}.
 \end{align*}
 Define the so-called \textit{paraorthogonal polynomials} as in Theorem 2.2.12 of \cite{S1} by $\Psi_{n+1}(z) := z\Phi_n(z) + \gamma_n\Phi_n^*(z)$, with $\Phi_n$ being the orthogonal polynomials corresponding to the CMV matrix with Verblunsky coefficients $\Theta\beta_j(\omega)$. Then Theorem 2.2.13 of \cite{S1} guarantees that the zeros of $\Psi_n$ are precisely $\set{z_i^{(n), r}}_{i = 1}^{n+1}$. On the other hand, if we set $d\nu_n^{\gamma_n}(\omega)$ to denote the counting probability measure on the zeros of $\Psi_n$, then Theorem 10.5.21 of \cite{S2} guarantees that the density of states measure $d\nu$ exists and $d\nu_n^{\gamma_n}(\omega)\rightarrow d\nu$ weakly. On the other hand, since $\partial \mathbb{D}$ is compact, and since the roots $\set{z_i^{(n)}}$ and $\set{z_i^{(n), r}}$ are all simple (and interlace in the sense that between any two of the former there is one of the latter, and vice versa), it is easy to see that for any continuous function $f$ on $\partial \mathbb{D}$,
 \begin{align*}
  \lim_{n\rightarrow\infty}\left(\int_{\partial \mathbb{D}}f d\nu_n^{\gamma_n}(\omega) - \int_{\partial\mathbb{D}}fd\nu_n(\omega)\right) = 0.
 \end{align*}
 Thus, since $d\nu_n^{\gamma_n}(\omega)\rightarrow d\nu$ weakly, we have $d\nu_n(\omega)\rightarrow d\nu$ weakly.
\end{proof}

Theorem \ref{thm:IDS-Ising-general} can be improved when the essential support of the spectral measure of the CMV matrix in question is not all of $\partial\mathbb{D}$, and when the system $(\omega, \mu, g, T)$ is uniquely ergodic (which is, for example, precisely the case with the Fibonacci substitution, where $\Omega$ is the hull, $T$ is the left shift on the hull, and $\mu$ is the unique ergodic measure):

\begin{theorem}\label{thm:IDS-Ising-improved}
 If $(\Omega, \mu, g, T)$ in Theorem \ref{thm:IDS-Ising-general} is uniquely ergodic, and if for each $\omega\in\Omega$, the essential support of the spectral measures of the CMV matrix is not all of $\partial\mathbb{D}$, then the conclusion of Theorem \ref{thm:IDS-Ising-general} holds for all $\omega\in\Omega$.
\end{theorem}

\begin{proof}
 If $\mathcal{C}(\omega)$ denotes the CMV matrix with Verblunsky coefficients given by $\set{\Theta\beta_j(\omega)}_{j\in\N}$, with $\omega\in\Omega$, let $\mathcal{C}^{(n)}(\omega)$ denote the so-called \textit{cutoff CMV matrix} given by the upper left $n\times n$ block of $\mathcal{C}(\omega)$. Let $d\eta_n(\omega)$ denote the counting probability measure on the eigenvalues of $\mathcal{C}^{(n)}(\omega)$. It is known that if the weak limit of the measures $d\eta_n(\omega)$ exists, and if the essential support of the spectral measures of $\mathcal{C}(\omega)$ is not all of $\partial\mathbb{D}$, then $d\eta_n^{\gamma_n}(\omega)$ from Theorem \ref{thm:IDS-Ising-general} (for any choice of $\gamma_n$) and $d\eta_n(\omega)$ have the same weak limit, supported on the spectrum of $\mathcal{C}(\omega)$ (for this, see Theorem~8.2.7 in \cite{S2}; see also the discussion preceding the statement of Theorem~3.4 in \cite{S3}). On the other hand, we know from Theorem \ref{thm:IDS-Ising-general} that the weak limit of $d
\nu_n^{\gamma_n}(\omega)$ exists for almost all $\omega$ and is independent of $\omega$, and from Theorem~10.5.21 from \cite{S2} we know that the limit of $d\eta_n(\omega)$ also exists for almost every $\omega\in\Omega$ and is $\omega$-independent, and coincides with the weak limit of $d\nu_n^{(\gamma_n)}(\omega)$. Thus it remains to prove that in fact $d\eta_n(\omega)\rightarrow d\nu$ weakly (with $d\nu$ from Theorem \ref{thm:IDS-Ising-general}) for \textit{all} $\omega\in\Omega$. For this we employ the argument of Hof from the proof of Proposition~7.2 in \cite{H} verbatim, employing unique ergodicity.
\end{proof}

\begin{appendix}

\section{Unitary Dynamics and Subordinacy Theory}

In this appendix we describe an adaptation of the Guarneri-Combes estimate, and of an extension thereof due to Last, to the setting relevant to quantum walks, namely discrete time dynamics generated by the iteration of a unitary operator on $\mathcal{H} = \ell^2(\Z) \otimes \C^2 \cong \ell^2(\Z)$, which has been worked out in \cite{DFV}. These estimates derive a lower bound for the spreading rate of a quantum walk in terms of suitable regularity properties of the spectral measure of the initial state relative to the unitary operator. We also describe an adaptation of the Damanik-Killip-Lenz theorem, which derives these regularity properties of spectral measures from suitable transfer matrix estimates, and which has been worked out in \cite{MO}.

Recall that we relabeled the basis elements and wrote them as $\{ e_m \}_{m \in \Z}$. It is the index of the latter basis elements that we refer to when we refer to position in space. Of course, this is closely and explicitly related to the index of the elements of the original basis.

We consider a non-zero finitely supported initial state $\psi \in \ell^2(\Z)$ and study the spreading in space of $U^n \psi$ as $|n| \to \infty$. Since the unitary operators $U$ we study are local in the sense that $\langle e_m , U e_k \rangle = 0$ if $|m-k| > 2$, it follows that for each $n \in \Z$, $U \psi$ is finitely supported as well. In particular, we can consider moments
$$
M_\psi(n,p) = \sum_{m \in \Z} (1 + |m|^p) |\langle e_m , U^n \psi \rangle|^2
$$
for $p > 0$, which are finite for all $n \in \Z$. For a general (not necessarily local) unitary operator, we implicitly assume that the moments we consider are finite, so that the questions we address are meaningful. For infinite moments, all the estimates given below hold trivially.

Next we average in time and write
$$
\tilde M_\psi(N,p) = \frac{1}{N} \sum_{n=0}^{N-1} M(n,p).
$$
We ask at what power-law rate these quantities grow as $N \to \infty$ and hence set
$$
\tilde \beta^+_\psi(p) = \limsup_{N \to \infty} \frac{\log \tilde M_\psi(N,p)}{p \log N}, \quad \tilde \beta^-_\psi(p) = \liminf_{N \to \infty} \frac{\log \tilde M_\psi(N,p)}{p \log N}.
$$
These quantities are called transport exponents. Our goal is to bound these transport exponents from below in terms of suitable regularity properties of the spectral measure associated with the initial state $\psi$ (and the operator $U$). Below we state a version of the Guarneri-Combes estimate, which is the simplest lower transport bound of this kind. The required regularity is formulated in the following definition.

\begin{definition}
Let $\alpha > 0$. A finite measure $\nu$ on the unit circle $\partial \D$ is called uniformly $\alpha$-H\"older continuous, denoted $U\alpha H$, if there is a constant $C < \infty$ such that for every arc $I \subseteq \partial \D$, we have $\nu(I) \le C |I|^\alpha$, where $|\cdot |$ denotes arc length.
\end{definition}

We can now state the Guarneri-Combes estimate for the setting we consider. It is an analog of the estimate Guarneri \cite{G} and Combes \cite{C} proved for the time-averaged Schr\"odinger evolution and the proof of this proposition can be found in \cite{DFV}.

\begin{prop}[Guarneri-Combes Estimate]\label{p.gcest}
Suppose the spectral measure associated with $U$ and $\psi$ is $U\alpha H$ for some $\alpha > 0$. Then, for every $p > 0$, there is a constant $C_p$ such that for every $N \in \Z_+$, we have
$$
\tilde M_\psi(N,p) \ge C_p N^{p\alpha}.
$$
In particular, for every $p > 0$, we have
$$
\tilde \beta^\pm_\psi(p) \ge \alpha.
$$
\end{prop}

It was observed by Last \cite{L} in the Schr\"odinger evolution context that by approximation with uniformly $\alpha$-H\"older continuous measures, lower bounds of this kind can be derived under a weaker assumption. The unitary analog of this result was also derived in \cite{DFV}, and we state this result in Proposition~\ref{p.gclest} below.

\begin{definition}
Let $\nu$ be a finite measure on $\partial \D$. For $z \in \partial \D$, the lower scaling exponent of $\nu$ at $z$ is given by
$$
s_\nu^-(z) = \liminf_{\varepsilon \downarrow 0} \frac{\log \nu(I(z,\varepsilon))}{\log \varepsilon},
$$
where $I(z,\varepsilon)$ is the arc of length $\varepsilon > 0$ centered at $z$. The Hausdorff dimension of $\nu$ is given by
$$
\dim_H(\nu) = \text{\rm $\mu$-ess-sup} \; s_\nu^-(z).
$$
\end{definition}

With this definition, the following was shown in  \cite{DFV}.

\begin{prop}[Guarneri-Combes-Last Estimate]\label{p.gclest}
Denote the spectral measure associated with $U$ and $\psi$ by $\mu_{U,\psi}$. Then, for every $p > 0$, we have
$$
\tilde \beta^\pm_\psi(p) \ge \dim_H(\mu_{U,\psi}).
$$
\end{prop}

\bigskip

Let us now assume that the unitary operator is in fact an extended CMV matrix, as was the case in the setting of Section~\ref{s.qw}, and to emphasize this, we re-denote the unitary operator by $\mathcal{E}$. In order to apply Proposition~\ref{p.gcest} or Proposition~\ref{p.gclest} (which hold for general unitary operators), we need a sufficient condition for uniform $\alpha$-H\"older continuity or a lower bound for the Hausdorff dimension of the spectral measure that applies to extended CMV matrices, and which we are able to check in our scenario.

As shown in \cite{MO}, convenient sufficient conditions can be formulated in terms of estimates for sequences vector-valued sequences obeying the recursion
\begin{equation}\label{e.tmequ}
\begin{pmatrix} \xi_{n+1} \\ \zeta_{n+1} \end{pmatrix} = T(z,\alpha_n) \begin{pmatrix} \xi_n \\ \zeta_n\end{pmatrix},
\end{equation}
where
\begin{equation}\label{e.tmequnorm}
|\xi_0| = |\zeta_0| = 1.
\end{equation}
The following result was shown in \cite{MO}. It is an adaptation of a result shown by Damanik, Killip, and Lenz in the Schr\"odinger context \cite{DKL}.

\begin{prop}[Damanik-Killip-Lenz Estimate]\label{p.dklest}
Suppose that for $z \in \partial \D$, there are constants $\gamma_1(z),\gamma_2(z), C_1(z), C_2(z) > 0$ so that
$$
C_1(z) L^{\gamma_1(z)} \leq \norm{\xi}_L \leq C_2(z) L^{\gamma_2(z)}
$$
for every solution of \eqref{e.tmequ} that is normalized in the sense of \eqref{e.tmequnorm}. Then, for every spectral measure of $\mathcal{E}$, the lower scaling exponent is bounded from below by $\frac{2\gamma_1(z)}{\gamma_1(z) + \gamma_2(z)}$.

In particular, if $S \subset \partial \D$ is a Borel set such that there are constants $\gamma_1,\gamma_2$ and, for each $z \in S$, there are constants $C_1(z), C_2(z)$ so that
$$
C_1(z) L^{\gamma_1} \leq \norm{\xi}_L \leq C_2(z) L^{\gamma_2}
$$
for every $z \in S$ and for every solution of \eqref{e.tmequ} that is normalized in the sense of \eqref{e.tmequnorm}. Then, the restriction of every spectral measure of $\mathcal{E}$ to $S$ is purely $\frac{2\gamma_1}{\gamma_1 + \gamma_2}$-continuous, that is, it gives zero weight to sets of zero $h^{\frac{2\gamma_1}{\gamma_1 + \gamma_2}}$ measure.

\end{prop}

\end{appendix}


\end{document}